\documentclass[1p,final]{elsarticle}
\usepackage{amsfonts,natbib,pslatex}

\usepackage{latexsym,amsmath,amsthm,amssymb,amsxtra,mathrsfs,times,CJK}
\usepackage{makecell}
\usepackage{color}
\usepackage{amsfonts,color}
\usepackage{amssymb,amsmath,latexsym}
\usepackage{amssymb}
\usepackage[colorlinks,linkcolor=blue,anchorcolor=blue,citecolor=blue]{hyperref}

\usepackage{threeparttable}

\newtheorem{theorem}{Theorem}[section]
\newtheorem{lemma}[theorem]{Lemma}

\newcommand{\GF}{{\mathrm {GF}}}

\begin{document}

\begin{frontmatter}



\title{Two Classes of Power Mappings with Boomerang Uniformity 2}

\author[SWJTU]{Zhen Li}
 \ead{lz-math@my.swjtu.edu.cn}
\author[SWJTU]{Haode Yan\corref{cor1}}
 \ead{hdyan@swjtu.edu.cn}

 \cortext[cor1]{Corresponding author}
 \address[SWJTU]{School of Mathematics, Southwest Jiaotong University, Chengdu, 610031, China}


\begin{abstract}
Let $q$ be an odd prime power. Let $F_1(x)=x^{d_1}$ and $F_2(x)=x^{d_2}$ be power mappings over $\GF(q^2)$, where $d_1=q-1$ and $d_2=d_1+\frac{q^2-1}{2}=\frac{(q-1)(q+3)}{2}$.
In this paper, we study the the boomerang uniformity of $F_1$ and $F_2$ via their differential properties. It is shown that, the boomerang uniformity of $F_i$ ($i=1,2$) is 2 with some conditions on $q$.
\end{abstract}


\end{frontmatter}

\section{Introduction}
Differential cryptanalysis \cite{BS1990} is one of the most fundamental analysis method to attack block cipher, which has attracted extensive researches for 30 years. In order to measure the ability of a given function to resist differential attack, Nyberg introduced the concept of differential uniformity in \cite{N1993}. Let $\GF(q)$ denote the finite field containing $q$ elements and $\GF(q)^*=\GF(q)\setminus\{0\}$. For a function $F$ from $\GF(q)$ to itself, the differential uniformity of $F$ is
$$
\Delta_F=\max_{a\in\GF(q)^*}\max_{b\in\GF(q)}\delta_F(a,b),
$$
where $\delta_F(a,b)=\#\{x\in\GF(q)~|~F(x+a)-F(x)=b\}$.
The lower the quantity of $\Delta_F$ is, the stronger the ability of function $F$ resisting differential attack. If $\Delta_F=1$, then $F$ is called perfect nonlinear (PN). Known results on PN functions were presented in \cite{BH2008,CM1997,DY2006,ZKW2009}. If $\Delta_F=2$, then $F$ is called almost perfect nonlinear (APN). For the known results on APN functions, the readers are referred to \cite{BDMW2010,BH2008,Dobbertin-welchcase,Dobbertin-nihocase,DMMPW2003,HRS1999,HS1997,NH2007,N1993,ZW2011}.

Power mappings (i.e., monomials) with low differential uniformity serve as
good candidates for the design of S-boxes not only because of
their strong resistance to differential attacks but also for the
usually low implementation cost in hardware. For any power mappings $F(x)=x^d$, we have $\delta_F(a,b)=\delta_F(1,\frac b{a^d})$ for arbitrary $(a,b)\in\GF(q)^*\times\GF(q)$. Hence the differential properties of the power function $F$ are completely determined by $\delta_F(1,b)$ when $b$ varies through $\GF(q)$. In \cite{BCC2010}, Blondeau, Canteaut and Charpin defined the differential spectrum of a power function. Let $\omega_i=\#\{b\in\GF(q)~|~\delta_F(1,b)=i\}$ for $0\leq i\leq\Delta_F$, where $\Delta_F$ is the differential uniformity of $F$. The differential spectrum of $F$ is defined as the multiset
$$
\mathbb S=\{\omega_i~|~0\leq i\leq\Delta_F, \omega_i>0\}.
$$
 The differential spectrum is an important concept of cryptographic functions. The differential spectrum of any PN function is $\mathbb S=\{\omega_1=q\}$, while the differential spectrum of any APN function over even characteristic finite field $\GF(2^n)$ is $\mathbb S=\{\omega_0=\omega_2=2^{n-1}\}$. For the results of power mappings with known differential spectra, the readers are referred to \cite{BCC2010,BCC2011,CHNC2013,Dobbertin2001,L.2020,LY2021,Xia2020,XY2017,XYY2018,YL2021,YXLHXL,YZWWHW2019} and their references for more information.
 In 1999, Wagner introduced a new cryptanalysis method against block ciphers, namely, the boomerang attack \cite{W1999}. It can be regard as a generalization of the differential attack, and it allows new avenues of attack for many ciphers previously deemed safe from differential cryptanalysis. Nowadays, the boomerang distinguisher is one of the most powerful distinguishers for identifying weaknesses in block ciphers. The boomerang distinguisher has been recently deemed another highly important distinguisher of block ciphers. In order to study this attack method in detailed, Cid \textit{et al.} introduced the Boomerang Connectivity Table (BCT) in EUROCRYPT 2018 \cite{CTTYL2018}. To quantify the resistance of a function against the boomerang attack, Boura and Canteaut introduced the concept of boomerang uniformity \cite{BC2018}. Later in \cite{LQSL2019}, Li, Qu, Sun and Li generalized the definition of the boomerang uniformity. Let $F$ be a function defined on $\GF(p^n)$. For any $a,b\in\GF(p^n)^*$, denote by $\beta_F(a,b)$ the number of the solutions $(x,y)$ in $\GF(p^n)\times \GF(p^n)$ of the equation system
\begin{align*}
    \left\{
      \begin{array}{ll}
   F(x)-F(y) & =b, \\
   F(x+a)-F(y+a) & =b.
      \end{array}
    \right.
\end{align*}
The boomerang uniformity of $F$ is defined as
$$
\beta_F=\max_{a,b\in \GF(p^n)^*}\{\beta_F(a,b)\}.
$$
The lower the quantity of $\beta_F$ is, the stronger the ability of function $F$ resisting boomerang attack. Recent progress on cryptographic functions with known boomerang uniformity can refer to the survey written by Mesnager, Mandal and Msahli \cite{MMM}.



The boomerang uniformity of power mappings attracts lots of attention. When $F(x)=x^d$ is a power function, we have $\beta_F=\max\limits_{b\in \GF(p^n)^*}\{\beta_F(1,b)\}$ similarly. The known results are introduced as follows. For power functions over even characteristic finite fields $\GF(2^n)$, the boomerang uniformity of the inverse function $x^{2^n-2}$ and the Gold function $x^{2^k+1}$ were presented in \cite{BC2018}. Zha and Hu studied the boomerang uniformity of power permutations $x^{2^k-1}$ in \cite{ZH2019}.
Calderini and Villa studied the Bracken-Leader function $x^{2^{2k}+2^k+1}$, where $n=4k$ and $k$ is odd. They showed that the boomerang uniformity of the Bracken-Leader function is upper bounded by 24 \cite{CV2020}. Recently, Hasan, Pal and Stanica studied the power mapping $x^{q-1}$ over $\GF(q^2)$, where $q=2^m$ is a power of $2$. It was shown that the boomerang uniformity is $2$ if $m$ is even, and is $4$ if $m$ is odd \cite{HPS2105.04284}. For odd characteristic finite fields, there are few results. In \cite{JLLQ2021}, Jiang \textit{et al.} determined the boomerang uniformity of $x^{p^n-2}$ over $\GF(p^n)$ and $x^{\frac{3^n+3}2}$ over $\GF(3^n)$ with odd $n$.

We focus on two power mappings over odd characteristic finite fields. Let $F_1(x)=x^{q-1}$ and $F_2(x)=x^{\frac{(q-1)(q+3)}{2}}$ be power mappings over $\GF(q^2)$. In this paper, we present that the boomerang uniformity of $F_i$ ($i=1,2$) is $2$. The rest of this paper is organized as follows. In Section \ref{sec2}, we introduce some notation and the differential properties of $F_1$ and $F_2$. In section \ref{sec3}, we investigate the boomerang uniformity of $F_1$ and $F_2$ via their differential properties. Section \ref{sec4} concludes this paper.

\section{Preliminaries}\label{sec2}
We begin this section by fixing some notation which will be used throughout this paper unless otherwise stated.
\begin{itemize}
	\item $q$ is an odd prime power.
	\item Let $F_1(x)=x^{d_1}$ and $F_2(x)=x^{d_2}$ be power mappings over $\GF(q^2)$, where $d_1=q-1$ and $d_2=d_1+\frac{q^2-1}{2}=\frac{(q-1)(q+3)}2$.
	\item $\Delta_i(x)=F_i(x+1)-F_i(x)=(x+1)^{d_i}-x^{d_i}$, where $i=1,2$.
	\item For any $b \in \GF(q^2)$, let  $\Delta^{-1}_i(b)=\{x~|~\Delta_i(x)=b\}$ and $\delta_i(b)=\#\Delta^{-1}_i(b)$, where $i=1,2$.
\end{itemize}
In the following, we give the differential properties of $F_1$ and $F_2$ with some conditions on $q$, which will be used in the sequel.

\begin{lemma}\label{sec2-lm1}Let $q\not\equiv2\pmod3$. We have,
  \begin{itemize}
    \item $\Delta^{-1}_1(0)=\GF(q)\setminus\{0,-1\}$, $\delta_1(0)=q-2$;
    \item $\Delta^{-1}_1(1)=\{0\}, \Delta^{-1}_1(-1)=\{-1\}, \delta_1(1)=\delta_1(-1)=1$;
    \item $\delta_1(b)\leq2$ for all $b\neq0,\pm1$.
  \end{itemize}
\end{lemma}

\begin{proof}

We mainly study the derivative equation
\begin{equation}\label{sec2-1}
  \Delta_1(x)=(x+1)^{q-1}-x^{q-1}=b
\end{equation}
for arbitrary but fixed $b\in\GF(q^2)$.
Clearly, $\Delta_1(0)=1$ and $\Delta_1(-1)=-1$. We assume that $b=0$ first, then $x\neq0,-1$. By (\ref{sec2-1}), we obtain
\begin{equation*}\label{sec2-2}
  (1+\frac1x)^{q-1}=1,
\end{equation*}
then $x\in\GF(q)$.
{Moreover, any $ x\in\GF(q)\setminus\{0,-1\}$ gives $\Delta_1(x)=0$.}
We conclude that $\Delta^{-1}_1(0)=\GF(q)\setminus\{0,-1\}$ and  $\delta_1(0)=q-2$.

Next we assume that $b\neq0$. Multiplying
  $x(x+1)$ on each side of (\ref{sec2-1}) gives
  \begin{equation}\label{sec2-2'}
    x-bx(x+1)=x^{q}.
  \end{equation}
Then $x^q-b^qx^q(x^q+1)=x$. This with (\ref{sec2-2'}) leads to
  \begin{equation}\label{sec2-2''}
    x^2+(1-\frac{2}{b})x+\frac{b^{q-1}+1-b^{q}}{b^{q+1}}=0,
  \end{equation}
  which is a quadratic equation of $x$. Therefore, $\delta(b)\leq2$ for all $b\neq0,\pm1$.

When $b=1$, it is obvious that $x=0$ is a solution of (\ref{sec2-1}). When $x\neq0,-1$,  the equation (\ref{sec2-2'}) becomes
\begin{equation*}\label{sec2-3}
  x^{q-2}=-1.
\end{equation*}
Note that $\gcd(2(q-2),q^2-1)=2$ since $q\not\equiv2\pmod3$. We obtain $x=-1$, which is not considered. Then $\Delta^{-1}(1)=\{0\}$ and $\delta_1(1)=1$.

Since $q-1$ is even, $x\in\Delta^{-1}_1(b)$ if and only if $-x-1\in\Delta^{-1}_1(-b)$, then $\delta_1(b)=\delta_1(-b)$ for any $b\in\GF(q^2)$. The desired results follows.
\end{proof}

\begin{lemma}\label{sec3-lm1}
	Let $q\not\equiv2\pmod3$ and $q\equiv3\pmod4$. We have
	\begin{itemize}
		
			\item $\Delta^{-1}_2(0)=\GF(q)\setminus\{0,-1\}$, $\delta_2(0)=q-2$;
			\item $\Delta^{-1}_2(1)=\{0\}, \Delta^{-1}_2(-1)=\{-1\}, \delta_2(1)=\delta_2(-1)=1$;
			\item $\delta_2(b)\leq2$ for all $b\neq0,\pm1$.
	\end{itemize}
\end{lemma}

\begin{proof}
We mainly study the derivative equation
\begin{equation}\label{sec3-eq1}
	\Delta_2(x)=(x+1)^{d_2}-x^{d_2}=b
\end{equation}
for arbitrary but fixed $b\in\GF(q^2)$, where $d_2=\frac{(q-1)(q+3)}2$ is even. Clearly, $\Delta_2(0)=1$ and $\Delta_2(-1)=-1$. When $b=0$, 	 $x\neq0,-1$. By (\ref{sec3-eq1}), we obtain
\begin{equation*}
	(1+\frac1x)^{d_2}=1,
\end{equation*}
then $x\in\GF(q)$ since $\gcd(d_2,q^2-1)=q-1$. We conclude that $\Delta^{-1}_2(0)=\GF(q)\setminus\{0,-1\}$ and  $\delta_2(0)=q-2$.

Next we assume that $b\neq0$. Let $\chi(x)=x^{\frac{q^2-1}2}$ be the quadratic multiplicative character of $\GF(q^2)^*$. For $x\neq0,-1$, we distinguish the following two cases.

{\bf Case I.} $\chi(x+1)=\chi(x)$. Then (\ref{sec3-eq1}) becomes
\begin{equation}\label{case1}
	(x+1)^{q-1}-x^{q-1}=b\chi(x).
\end{equation}
Moreover,
\begin{equation}\label{chib1}
	b\chi(x)x(x+1)=x-x^q.
\end{equation}
Note that $(x-x^q)^q=-(x-x^q)$, $(x-x^q)^{q-1}=-1$, then
$\chi(x-x^q)=(x-x^q)^{\frac{(q-1)(q+1)}{2}}=1$ since $q\equiv3\pmod4$. From (\ref{chib1}) we conclude that if this case contributes solutions, then $\chi(b)=1$.

By Lemma \ref{sec2-lm1}, (\ref{case1}) has at most two solutions for each value of $\chi(x)$.
Note that $x=x_0$ is a solution of $(x+1)^{q-1}-x^{q-1}=b$ if and only if $x=-x_0-1$ is a solution of $(x+1)^{q-1}-x^{q-1}=-b$. However,  $\chi(-x_0-1)=\chi(x_0+1)=\chi(x_0)$. This implies that (\ref{case1}) cannot have solution for both two values of $\chi(x)$ simultaneously. This case contributes at most two solutions.


{\bf Case II.} $\chi(x+1)=-\chi(x)$. Then (\ref{sec3-eq1}) becomes
		\begin{equation}\label{case2}
 		(x+1)^{q-1}+x^{q-1}=-b\chi(x).
		\end{equation}
		Multiplying $x(x+1)$ on both sides of (\ref{case2}), we have
		\begin{equation}\label{sec3-eq2}
			2x^{q+1}+x^{q}+x=-b\chi(x)x(x+1).
		\end{equation}
		Note that $-b\chi(x)x(x+1)\in\GF(q)$ and $\chi(x(x+1))=-1$, we conclude that, if this case contributes solutions, then $\chi(b)=-1$. Moreover, we obtain
		\begin{equation}\label{case2xb}
			x^{q}=\frac{-b\chi(x)x(x+1)-x}{1+2x}
		\end{equation}
		 from (\ref{sec3-eq2}).
		From (\ref{case2xb}), we obtain that
		\begin{equation}\label{sec3-eq3}
			x^2+x+\frac{b^{q+1}+\chi(x)(b^q+b)}{-\chi(x)b(b^{q+1}-4)}=0.
		\end{equation}
		We mention that $b^{q+1}\neq4$. Otherwise, $b^{q+1}=4$, then $b^{\frac{q^2-1}{2}}=4^{\frac{q-1}{2}}=2^{q-1}=1$, which contradicts $\chi(b)=-1$. Note that both $b^{q+1}+\chi(x)(b^q+b)$ and $b^{q+1}-4$ are in $\GF(q)$.  For a given $\chi(x)$, the equation (\ref{sec3-eq3}) is a quadratic equation of $x$, the discriminant is $\frac{(b+2\chi(x))^2b^{q-1}}{b^{q+1}-4}$, which is a square element in $\GF(q^2)$. Let $x_1$ and $x_2$ be solutions of (\ref{sec3-eq3}),
		then $\chi(x_1x_2)=\chi(\frac{b^{q+1}+\chi(x)(b^q+b)}{-\chi(x)b(b^{q+1}-4)})=-1$ since $\chi(b)=-1$. Only one of $x_1$ and $x_2$ is a square element, and the other one is a nonsquare element. For each value of $\chi(x)$, there is at most one solution. This case contributes at most two solutions.
		
		For $b\neq0,\pm1$, the solutions are all in the above two cases. According to the value of $\chi(b)$, we know  that Cases I and II cannot have solution simultaneously. Thus we obtain $\delta_2(b)\leq 2$ for $b\neq0,\pm1$.
		
		At last, we consider $b=\pm 1$. $\Delta_2(0)=1$ and $\Delta_2(-1)=-1$ are obvious. It can be checked that (\ref{sec3-eq1}) has no solution in the two cases. Then $\Delta^{-1}_2(1)=\{0\}, \Delta^{-1}_2(-1)=\{-1\}$, and $\delta_2(1)=\delta_2(-1)=1$. We complete the proof.	
\end{proof}

\section{The Boomerang Uniformity of $F_1$ and $F_2$}\label{sec3}

In this section, we determine the boomerang uniformity of $F_1(x)=x^{q-1}$ and $F_2(x)=x^{\frac{(q-1)(q+3)}{2}}$ over $\GF(q^2)$. We have the following theorems.

\begin{theorem}\label{sec2-thm1} Let $q\not\equiv2\pmod3$.
The boomerang uniformity of $F_1(x)=x^{q-1}$ over $\GF(q^2)$ is 2.
\end{theorem}

\begin{proof}
To determine the boomerang uniformity of $F_1$, we mainly study the equation system
\begin{align}\label{sec2.2-1}
     \left\{
      \begin{array}{ll}
      x^{q-1}-y^{q-1}&=b\\
      (x+1)^{q-1}-(y+1)^{q-1}&=b
      \end{array}
    \right.
\end{align}
for arbitrary but fixed $b\in\GF(q^2)^*$.
It is obvious that the equation system (\ref{sec2.2-1}) is equivalent to the following one,
\begin{align}\label{sec2.2-1'}
     \left\{
      \begin{array}{ll}
      x^{q-1}-y^{q-1}&=b\\
      (x+1)^{q-1}-(y+1)^{q-1}&=b\\
      \Delta_1(x)&=\Delta_1(y),
      \end{array}
    \right.
\end{align}
where $\Delta_1(z)=(z+1)^{q-1}-z^{q-1}$ we defined before. Let $\Delta_1(x)=\Delta_1(y)=c$ for some $c\in\GF(q^2)$. By Lemma \ref{sec2-lm1}, we know that the equation $\Delta_1(x)=c$ (respectively, $\Delta_1(y)=c$) has $0,1,2$ or $q-2$ roots. More precisely, define
$$
\Omega_i=\{c\in\GF(q^2)~|~\delta_1(c)=i~\}
$$
for $i=0,1,2,q-2$. We mention that $\Omega_{q-2}=\{0\}$ and $\Omega_0\cup\Omega_1\cup\Omega_2=\GF(q^2)^*$. When $c$ runs through $\GF(q^2)$, we discuss in the following four cases.

{\bf Case 1.} $c\in\Omega_0$. By the definition of $\Omega_0$, $\Delta_1(x)=c$ (respectively, $\Delta_1(y)=c$) has no solution in $\GF(q^2)$. Hence this case cannot contribute solution to (\ref{sec2.2-1'}).

{\bf Case 2.} $c\in\Omega_1$. Let $\Delta_1^{-1}(c)=\{z_0\}$. In this case, $\Delta_1(x)=\Delta_1(y)=c$ has one solution $(x,y)=(z_0,z_0)$. This case cannot contribute solution to (\ref{sec2.2-1'}) since $b\neq 0$.

{\bf Case 3.} $c\in\Omega_{q-2}$, i.e., $c=0$. By Lemma \ref{sec2-lm1}, $\Delta_1^{-1}(0)=\GF(q)\setminus\{0,-1\}$. When $\Delta_1(x)=\Delta_1(y)=0$, i.e., $x,y\in\GF(q)\setminus\{0,-1\}$, we have $b=x^{q-1}-y^{q-1}=0$, which is a contradiction.
 Hence this case cannot contribute solution to (\ref{sec2.2-1'}).

{\bf Case 4.} $c\in\Omega_2$. We first consider the equation system (\ref{sec2.2-1'}). We know that $x,y\neq0,-1$. From the first two equations of  (\ref{sec2.2-1'}), we obtain
  \[x^qy-xy^q=bxy\]
  and
  \[(x+1)^q(y+1)-(x+1)(y+1)^q=b(x+1)(y+1).\]
Then
\[b(x+y+1)=b(x+1)(y+1)-bxy=x^q+y-x-y^q.\]
We have
\begin{equation}\label{1bxy}
b(x+y+1)+(b(x+y+1))^q=0.
\end{equation}

Let  $z_1$ and $z_2$ be the two distinct solutions of $\Delta_1(z)=c$. Then the possible solutions of (\ref{sec2.2-1'}) are
$(x,y)=(z_1,z_2)$ and $(x,y)=(z_2,z_1)$
since $b\neq 0$. Note that $z^{q-1}_1-z^{q-1}_2=-(z^{q-1}_2-z^{q-1}_1)$, $(z_1,z_2)$ and $(z_2,z_1)$ cannot be solutions of  (\ref{sec2.2-1'}) simultaneously. Without loss of generalization, we assume that $(z_1,z_2)$ is a solution of (\ref{sec2.2-1'}) for such $c$ and the fixed $b$.
 By the proof of Lemma \ref{sec2-lm1}, equation (\ref{sec2-2'}) shows that $z^{q-1}_i=1-c(z_i+1)$ for $i=1,2$. We have
\begin{equation}\label{1bc}
	b=z^{q-1}_1-z^{q-1}_2=c(z_1-z_2).
\end{equation}
The equation (\ref{sec2-2''}) shows that $z_1+z_2=-1+\frac{2}{c}$ and $z_1z_2=\frac{c^{q-1}+1-c^{q}}{c^{q+1}}$. Combining with (\ref{1bxy}) and (\ref{1bc}), we obtain the following relationship between $b$ and $c$.

\begin{align}\label{1relationbc}
	\left\{
	\begin{array}{ll}
		(\frac{b}{c})^{q-1}&=-1\\
		(\frac{b}{c})^2&=1-\frac{4}{c^{q+1}}.
	\end{array}
	\right.
\end{align}
For $c'\in\Omega_2$, $c'\neq c$, if $\Delta_1(x)=\Delta_1(y)=c'$ contributes another solution of (\ref{sec2.2-1'}), we know that $c'$ and $b$ also satisfy (\ref{1relationbc}). Let $c'=\alpha c$ for some $1\neq\alpha\in\GF(q^2)$. Then we obtain $\alpha^{q-1}=1$ from the first equation of (\ref{1relationbc}). Moreover, from the second equation of (\ref{1relationbc}), we have
\[c^2(1-\frac{4}{c^{q+1}})=(\alpha c)^2(1-\frac{4}{(\alpha c)^{q+1}}).\]
This with $\alpha^{q-1}=1$ leads to $\alpha^2=1$, then $\alpha=-1$, i.e., $c'=-c$. Only $-c\in\Omega_2$ can contribute solution of (\ref{sec2.2-1'}) for the given $b$.
More precisely, when $z_1$ and $z_2$ are distinct solutions of $\Delta_1(z)=c$, the solutions of $\Delta_1(z)=-c$ are $-z_1-1$ and $-z_2-1$. One can easily check that, when  $(z_1, z_2)$ is a solution of (\ref{sec2.2-1'}), so is $(-z_2-1,-z_1-1)$. We conclude that, for all $b\in\GF(q^2)$, the maximum of the number of solutions of  (\ref{sec2.2-1}) is 2. We complete the proof.
\end{proof}
\begin{theorem}\label{sec3-thm}
Let $q\not\equiv2\pmod3$ and $q\equiv3\pmod4$. The boomerang uniformity of $F_2(x)=x^{d_2}$ over $\GF(q^2)$ is 2, where $d_2=\frac{(q-1)(q+3)}2$.
\end{theorem}

\begin{proof}
   We consider the equivalent equation system
  \begin{align}\label{sec3-eq4}
  \left\{
    \begin{array}{ll}
    x^{d_2}-y^{d_2}&=b\\
    (x+1)^{d_2}-(y+1)^{d_2}&=b\\
    \Delta_2(x)&=\Delta_2(y)\\
    \end{array}
  \right.
  \end{align}
  for arbitrary but fixed $b\in\GF(q^2)^*$, where $\Delta_2(z)=(z+1)^{d_2}-z^{d_2} $ we defined before. Similarly, we define
$$
\Omega_i=\{c\in\GF(q^2)~|~\delta_2(c)=i\}
$$
for $i=0,1,2,q-2$. We mention that $\Omega_{q-2}=\{0\}$ and $\Omega_0\cup\Omega_1\cup\Omega_2=\GF(q^2)^*$. Let $\Delta_2(x)=\Delta_2(y)=c$, we discuss the equation system (\ref{sec3-eq4}) when $c$ runs through $\GF(q^2)$.
Similar to the proof of Theorem \ref{sec2-thm1}, we can prove that when $c\in\Omega_0\cup\Omega_1\cup\Omega_{q-2}$, the corresponding pair $(x,y)$ which comes from $\Delta_2(x)=\Delta_2(y)=c$ is not a solution of (\ref{sec3-eq4}). We only consider the case that $c\in\Omega_2$. Let $z_1$ and $z_2$ be the solutions of $\Delta_2(z)=c$, similar to the proof of Theorem \ref{sec3-lm1}, possible solutions are $(z_1,z_2)$ and $ (z_2,z_1)$. We always have,
\begin{equation}\label{bz1z2}
b^2=(z_1^{d_2}-z_2^{d_2})^2=(\chi(z_1)z^{q-1}_1-\chi(z_2)z^{q-1}_2)^2.
\end{equation}
 According to Lemma \ref{sec3-lm1}, we discuss the relationship of $b$ and $c$ in the following two cases.

{\bf Case 1.} $z_1$ and $z_2$ belong to Case I of Lemma (\ref{sec3-lm1}), i.e., $\chi(z_1+1)=\chi(z_1)=\chi(z_2+1)=\chi(z_2)$. In this case, $z_1$ and $z_2$ are the two solutions of
\[(z+1)^{q-1}-z^{q-1}=c\chi(z). \]
For a given value of $\chi(z_1)$, by the proof of Lemma \ref{sec2-lm1}, we obtain $z^{q-1}_i=1-c\chi(z_1)(z_i+1)$  ($i=1,2$) from (\ref{sec2-2'}), and
$z_1+z_2=-1+\frac{2\chi(z_1)}{c}$, $z_1z_2=\frac{c^{q-1}+1-c^{q}\chi(z_1)}{c^{q+1}}$ from (\ref{sec2-2''}).
From (\ref{bz1z2}) we have
\begin{equation*}\label{sec3-eq9'}
  b^2=(\chi(z_1)z^{q-1}_1-\chi(z_2)z^{q-1}_2)^2=c^2(z_1-z_2)^2=c^2-\frac4{c^{q-1}}.
\end{equation*}


{\bf Case 2.}  $z_1$ and $z_2$ belong to Case II of Lemma (\ref{sec3-lm1}), i.e., $\chi(z_1+1)=\chi(z_2)=1,\chi(z_1)=\chi(z_2+1)=-1$ (do not consider the order of $z_1$ and $z_2$). Based on (\ref{sec3-eq3}), $z_1$ and $z_2$ satisfy of $z_1^2+z_1+\frac{c^{q+1}-c^q-c}{c(c^{q+1}-4)}=0$ and $z_2^2+z_2+\frac{c^{q+1}+c^q+c}{-c(c^{q+1}-4)}=0$, respectively.
We have

\[z_1 =\frac12(-1+\frac{(c-2)c^{\frac{q-1}2}}{\sqrt{c^{q+1}-4}}) ~~~\mathrm{or}~~~ \frac12(-1-\frac{(c-2)c^{\frac{q-1}2}}{\sqrt{c^{q+1}-4}}),\]
\[z_2 =\frac12(-1+\frac{(c+2)c^{\frac{q-1}2}}{\sqrt{c^{q+1}-4}}) ~~~\mathrm{or}~~~ \frac12(-1-\frac{(c+2)c^{\frac{q-1}2}}{\sqrt{c^{q+1}-4}}).\]
There are four possible pairs of $(z_1,z_2)$. Moreover, by (\ref{case2xb}) we have
$$
z_1^{q-1} = \frac{c(z_1+1)-1}{1+2z_1}
$$
and
$$
  z_2^{q-1} = \frac{-c(z_2+1)-1}{1+2z_2},
$$
then
\begin{equation}\label{sec3-eq11'}
b=\chi(z_1)z^{q-1}_1-\chi(z_2)z^{q-1}_2=\frac{c(z_2-z_1)+2(1+z_1+z_2)}{(1+2z_1)(1+2z_2)}.
\end{equation}
Plugging the possible pairs $(z_1,z_2)$ into (\ref{sec3-eq11'}), we obtain
\begin{equation*}\label{sec3-eq11''}
	b^2=c^2-\frac4{c^{q-1}},
\end{equation*}
when $(z_1,z_2)=(\frac12(-1+\frac{(c-2)c^{\frac{q-1}2}}{\sqrt{c^{q+1}-4}}), \frac12(-1-\frac{(c+2)c^{\frac{q-1}2}}{\sqrt{c^{q+1}-4}}))$ and $(z_1,z_2)=(\frac12(-1-\frac{(c-2)c^{\frac{q-1}2}}{\sqrt{c^{q+1}-4}}), \frac12(-1+\frac{(c+2)c^{\frac{q-1}2}}{\sqrt{c^{q+1}-4}}))$. The other two pairs give $b=0$ and we discard them.

In the above two cases, we always have $b^2=c^2-\frac4{c^{q-1}}$. Then $(\frac{b}{c})^2=1-\frac{4}{c^{q+1}}\in\GF(q)$. We obtain $(\frac{b}{c})^{q-1}=1$ or $-1$. In Case 1, we claim that $(\frac{b}{c})^{q-1}=-1$. The proof is the same as the determination of the first equation in (\ref{1relationbc}). In what follows, we prove that $(\frac{b}{c})^{q-1}=-1$ in Case 2. In Case 2, we have
\[c=(z_1+1)^{d_2}-z^{d_2}_1=(z_1+1)^{q-1}+z^{q-1}_1=\frac{2z^{q+1}_1+z^q_1+z_1}{z_1(z_1+1)},\]
and
\[c^{q+1}=(\frac{2z^{q+1}_1+z^q_1+z_1}{z_1(z_1+1)})^{q+1}=\frac{(2z^{q+1}_1+z^q_1+z_1)^{q+1}}{z^{q+1}_1(z_1+1)^{q+1}}.\]
Then
\[(\frac{b}{c})^2=1-\frac{4}{c^{q+1}}=1-\frac{4z^{q+1}_1(z_1+1)^{q+1}}{(2z^{q+1}_1+z^q_1+z_1)^{q+1}}=(\frac{z^q_1-z_1}{2z^{q+1}_1+z^q_1+z_1})^2.\]
Note that $2z^{q+1}_1+z^q_1+z_1\in\GF(q)$ and $z^q_1-z_1\notin\GF(q)$, we assert that $(\frac{b}{c})^{q-1}=-1$.
By the discussion above, we have
\begin{align*}
	\left\{
	\begin{array}{ll}
		(\frac{b}{c})^{q-1}&=-1\\
		(\frac{b}{c})^2&=1-\frac{4}{c^{q+1}},
	\end{array}
	\right.
\end{align*}
which is the same as the relationship shown in the proof of Theorem \ref{sec2-thm1}. Then we can similarly prove that, the maximum number of the solutions of (\ref{sec3-eq4}) for all $b\in\GF(q^2)$ is 2. The desired result follows.
\end{proof}

\section{Concluding and Remarks}\label{sec4}

In this paper, two classes of power mappings with boomerang uniformity 2 were presented. We used their differential properties to study their boomerang uniformity. Our approach can be used to study the boomerang uniformity of other power mappings. It would be interesting to find more infinite classes of monomials with low boomerang uniformity.


\end{document}